\documentclass[journal]{IEEEtran}

\ifCLASSINFOpdf
\else
   \usepackage[dvips]{graphicx}
\fi
\usepackage{url}

\usepackage{graphicx}
\usepackage{amssymb} 
\usepackage{epstopdf}
\usepackage[ruled,norelsize]{algorithm2e}
\makeatletter
\newcommand{\removelatexerror}{\let\@latex@error\@gobble}
\makeatother
\usepackage{pseudocode}
\usepackage{amsthm}
\usepackage{amsmath}
\usepackage{cite}
\usepackage{tikz}
\usepackage{marginnote}
\usepackage{subfigure}
\usepackage{hyperref}
\usepackage{mathtools}
\usepackage{floatrow}
\newtheorem{thm}{Theorem}

\newtheorem{prop}[thm]{Proposition}

\makeatletter
\newcommand*\bigcdot{\mathpalette\bigcdot@{.65}}
\newcommand*\bigcdot@[2]{\mathbin{\vcenter{\hbox{\scalebox{#2}{$\m@th#1\bullet$}}}}}
\makeatother

\usepackage{booktabs,array}

\newcount\rowc

\makeatletter
\def\ttabular{%
	\hbox\bgroup
	\let\\\cr
	\def\rulea{\ifnum\rowc=\@ne \hrule height 1.3pt \fi}
	\def\ruleb{
		\ifnum\rowc=1\hrule height 1.3pt \else
		\ifnum\rowc=6\hrule height \heavyrulewidth 
		\else \hrule height \lightrulewidth\fi\fi}
	\valign\bgroup
	\global\rowc\@ne
	\rulea
	\hbox to 10em{\strut \hfill##\hfill}%
	\ruleb
	&&%
	\global\advance\rowc\@ne
	\hbox to 10em{\strut\hfill##\hfill}%
	\ruleb
	\cr}
\def\endttabular{%
	\crcr\egroup\egroup}

\begin{document}

\title{Dynamic imaging using a deep generative SToRM (Gen-SToRM) model}

\author{Qing Zou, Abdul Haseeb Ahmed, Prashant Nagpal, Stanley Kruger, Mathews Jacob
\thanks{Qing Zou is with the Applied Mathematics and Computational Sciences (AMCS) program at the University of Iowa, Iowa City, USA (e-mail: zou-qing@uiowa.edu). Abdul Haseeb Ahmed and Mathews Jacob are with the Department of Electrical and Computer Engineering, University of Iowa, Iowa City, USA (e-mail: abdul-ahmed@uiowa.edu and mathews-jacob@uiowa.edu). Prashant Nagpal and Stanley Kruger are with the Department of Radiology, University of Iowa, Iowa City, USA (e-mail: prashant-nagpal@uiowa.edu and stanley-kruger@uiowa.edu). This work is supported by NIH under Grants R01EB019961 and R01AG067078-01A1. This work was conducted on an MRI instrument funded by 1S10OD025025-01.}}

\maketitle

\begin{abstract}
We introduce a generative smoothness regularization on manifolds (SToRM) model for the recovery of dynamic image data from highly undersampled measurements. The model assumes that the images in the dataset are non-linear mappings of low-dimensional latent vectors. We use the deep convolutional neural network (CNN) to represent the non-linear transformation. The parameters of the generator as well as the low-dimensional latent vectors are jointly estimated only from the undersampled measurements. This approach is different from traditional CNN approaches that require extensive fully sampled training data. We penalize the norm of the gradients of the non-linear mapping to constrain the manifold to be smooth, while temporal gradients of the latent vectors are penalized to obtain a smoothly varying time-series. The proposed scheme brings in the spatial regularization provided by the convolutional network. The main benefit of the proposed scheme is the improvement in image quality and the orders-of-magnitude reduction in memory demand compared to traditional manifold models. To minimize the computational complexity of the algorithm, we introduce an efficient progressive training-in-time approach and an approximate cost function. These approaches speed up the image reconstructions and offers better reconstruction performance.

\end{abstract}


\IEEEpeerreviewmaketitle

\section{Introduction}

\IEEEPARstart{T}{he} imaging of time-varying objects at high spatial and temporal resolution is key to several modalities, including MRI and microscopy. A central challenge is the need for high resolution in both space and time \cite{tsao2003k,lingala2011accelerated}. Several computational imaging strategies have been introduced in MRI to improve the resolution, especially in the context of free-breathing and ungated cardiac MRI. A popular approach pursued by several groups is self-gating, where cardiac and respiratory information is obtained from central k-space regions (navigators) using bandpass filtering or clustering \cite{feng2014golden,feng2016xd,christodoulou2018magnetic,prieto2015highly,bustin2020compressed}. The data is then binned to the respective phases and recovered using total variation or other priors. Recently, approaches using smooth manifold regularization have been introduced. These approaches model the images in the time series as points on a high-dimensional manifold \cite{poddar2015dynamic,ahmed2020free,poddar2019manifold,nakarmi2017kernel,nakarmi2018mls}. Manifold regularization algorithms, including the smoothness regularization on manifolds (SToRM) framework \cite{poddar2015dynamic,ahmed2020free,poddar2019manifold}, have shown good performance in several dynamic imaging applications. Since the data is not explicitly binned into specific phases as in the self-gating methods, manifold algorithms are less vulnerable to clustering errors than navigator-based corrections. Despite the benefits, a key challenge with the current manifold methods is the high memory demand. Unlike self-gating methods that only recover specific phases, manifold methods recover the entire time series. The limited memory on current GPUs restricts the number of frames that can be recovered simultaneously, which makes it challenging to extend the model to higher dimensionalities. The high memory demand also makes it difficult to use spatial regularization priors on the images using deep learned models.

Our main focus is to capitalize on the power of deep convolutional neural networks (CNN) to introduce a memory efficient generative or synthesis formulation of SToRM. 
CNN based approaches are now revolutionizing image reconstruction, offering significantly improved image quality and fast image recovery \cite{wang2016perspective,wang2017machine,dardikman2020learned,ye2018deep,jin2017deep,monga2019algorithm,pramanik2020deep}. In the context of MRI, several novel approaches have been introduced \cite{wang2016accelerating, wang2020deepcomplexmri}, including transfer-learning \cite{dar2020transfer}, domain adaptation \cite{han2018deep}, learning-based dynamic MRI \cite{sanchez2020scalable,wang2019dimension,wang2020lantern}, and generative-adversarial models \cite{dar2020prior,dar2019image,yurt2019mustgan}.
Unlike many CNN-based approaches, the proposed scheme does not require pre-training using large amounts of training data. This makes the approach desirable in free-breathing applications, where the acquisition of fully sampled training data is infeasible. We note that the classical SToRM approach can be viewed as an analysis regularization scheme (see Fig. \ref{illus}.(a)). Specifically, a non-linear injective mapping is applied on the images such that the mapped points of the alias-free images lie on a low-dimensional subspace \cite{poddar2019manifold,zou2019sampling,zou2020recovery}. When recovering images from undersampled data, the nuclear norm prior is applied in the transform domain to encourage their non-linear mappings to lie in a subspace. Unfortunately, this analysis approach requires the storage of all the image frames in the time series, which translates to high memory demand. The proposed generative SToRM formulation offers quite significant compression of the data, which can overcome the above challenge. Both the relation between the analysis and synthesis formulations and the relation of the synthesis formulation to neural networks were established in earlier work \cite{zou2020recovery}. 

We assume that the image volumes in the dataset are smooth non-linear functions of a few latent variables, i.e., $\mathbf{x}_t = \mathcal{G}_{\theta}(\mathbf{z}_t)$, where $\mathbf{z}_t$ are the latent vectors in a low-dimensional space. $\mathbf{x}_t$ is the $t$-th generated image frame in the time series. This explicit formulation implies that the image volumes lie on a smooth non-linear manifold in a high-dimensional ambient space (see Fig. \ref{illus}.(b)). The latent variables capture the differences between the images (e.g., cardiac phase, respiratory phase, contrast dynamics, subject motion). We model the $\mathcal G$ using a CNN, which offers a significantly compressed representation. Specifically, the number of parameters required by the model (CNN weights and latent vectors) are several orders of magnitude smaller than required for the direct representation of the images. The compact model proportionately reduces the number of measurements needed to recover the images. In addition, the compression also enables algorithms with much smaller memory footprint and computational complexity. We propose to jointly optimize for the network parameters $\theta$ and the latent vector $\mathbf{z}_t$ based on the given measurements. The smoothness of the manifold generated by $\mathcal{G}_\theta(\mathbf z)$ depends on the gradient of $\mathcal G_{\theta}$ with respect to its input. To enforce the learning of a smooth image manifold, we regularize the norm of the Jacobian of the mapping $\|J_{z} \mathcal G_{\theta}\|^2$. We experimentally observe that by penalizing the gradient of the mapping, the network is encouraged to learn meaningful mappings. Similarly, the images in the time series are expected to vary smoothly in time. Hence, we also use a Tikhonov smoothness penalty on the latent vectors $\mathbf z_t$ to further constrain the solutions. We use the ADAM optimizer with stochastic gradients, where random batches of $\mathbf z_i$ and $\mathbf b_i$ are chosen at iteration to determine the parameters. Unlike traditional CNN methods that are fast during testing/inference, the direct application of this scheme to the dynamic MRI setting is computationally expensive. We use approximations, including progressive-in-time optimization and an approximated data term that avoids non-uniform fast Fourier transforms, to significantly reduce the computational complexity of the algorithm. 

The proposed approach is inspired by deep image prior (DIP), which was introduced for static imaging problems \cite{ulyanov2018deep}, as well as its extension to dynamic imaging \cite{jin2019time}. The key difference of the proposed formulation is the joint optimization of the latent variables $\mathbf z$ and $\mathcal G$. The work of Jin ea tl. \cite{jin2019time} was originally developed for CINE MRI, where the latent variables were obtained by linearly interpolating noise variables at the first and last frames. Their extension to real-time applications involved setting noise latent vectors at multiples of a preselected period, followed by linearly interpolating the noise variables. This approach is not ideally suited for applications with free breathing, when the motion is not periodic. Another key distinction is the use of regularization priors on the network parameters and latent vectors, which encourages the mapping to be an isometry between latent and image spaces. Unlike DIP methods, the performance of the network does not significantly degrade with iterations. While we call our algorithm ``generative SToRM'', we note that our goal is not to generate random images from stochastic inputs as in generative-adversarial networks (GAN). In particular, we do not use adversarial loss functions where a discriminator is jointly learned as in the literature \cite{bora2018ambientgan,kazuhiro2018generative}.

\section{Background}
\subsection{Dynamic MRI from undersampled data: problem setup}
Our main focus is to recover a series of images $\mathbf x_1,..\mathbf x_M$ from their undersampled multichannel MRI measurements. The multidimensional dataset is often compactly represented by its Casoratti matrix 
\begin{equation}
\mathbf X = \begin{bmatrix}
\mathbf x_1 & ... &\mathbf x_M
\end{bmatrix}.
\end{equation}
Each of the images is acquired by different multichannel measurement operators
\begin{equation}\label{key}
\mathbf b_i = \mathcal A_i(\mathbf x_i) + \mathbf n_i,
\end{equation}
where $\mathbf n_i$ is zero mean Gaussian noise matrix that corrupts the measurements.

\subsection{Smooth manifold models for dynamic MRI}
The smooth manifold methods model images $\mathbf x_i$ in the dynamic time series as points on a smooth manifold $\mathcal M$. These methods are motivated by continuous domain formulations that recover a function $f$ on a manifold from its measurements as 
\begin{equation}\label{key}
f = \arg \min_{f} \sum_i\|f(\mathbf x_i)-\mathbf b_i\|^2 + \lambda \int_{\mathcal M}  \|\nabla_{\mathcal M} f\|^2 d\mathbf x
\end{equation}
where the regularization term involves the smoothness of the function on the manifold. 
This problem is adapted to the discrete setting to solve for images lying on a smooth manifold from its measurements as
\begin{equation}\label{storm1}
\mathbf X = \arg \min_{\mathbf X} \sum_{i=1}^M\|\mathcal A(\mathbf x_i) - \mathbf b_i\|^2 + \lambda~{\mathrm{trace}}(\mathbf{X}\mathbf{L}\mathbf{X}^H),
\end{equation}
where $\mathbf L$ is the graph Laplacian matrix. $\mathbf L$ is the discrete approximation of the Laplace-Beltrami operator on the manifold, which depends on the structure or geometry of the manifold. The manifold matrix $\mathbf L$ is estimated from k-space navigators. Different approaches, ranging from proximity-based methods \cite{poddar2015dynamic} to kernel low-rank regularization \cite{poddar2019manifold} and sparse optimization \cite{nakarmi2018mls}, have been introduced. 

The results of earlier work \cite{zou2019sampling,poddar2019manifold} show that the above manifold regularization penalties can be viewed as an analysis prior. In particular, these schemes rely on a fixed non-linear mapping $\varphi$ of the images. The theory shows that if the images $\mathbf x_1,..\mathbf x_M$ lie in a smooth manifold/surface or union of manifolds/surfaces, the mapped points live on a subspace or union of subspaces. The low-dimensional property of the mapped points $\varphi(\mathbf x_1),..\varphi(\mathbf x_M)$ is used to recover the images from undersampled data or derive the manifold using a kernel low-rank minimization scheme:
\begin{equation}\label{storm}
\mathbf X^* = \arg \min_{\mathbf X} \sum_{i=1}^M\|\mathcal A(\mathbf x_i) - \mathbf b_i\|^2  + \lambda~ \|\left[\varphi(\mathbf x_1),..,\varphi(\mathbf x_N)\right]\|_*.
\end{equation}
This nuclear norm regularization scheme is minimized using an iterative reweighted algorithm, whose intermediate steps match \eqref{storm1}. The non-linear mapping $\varphi$ may be viewed as an analysis operator that transforms the original images to a low-dimensional latent subspace, very similar to analysis sparsity-based approaches used in compressed sensing.

\subsection{Unsupervised learning using Deep Image Prior}
The recent work of DIP uses the structure of the network as a prior \cite{ulyanov2018deep}, enabling the recovery of images from ill-posed measurements without any training data. Specifically, DIP relies on the property that CNN architectures favor image data more than noise. The regularized reconstruction of an image from undersampled and noisy measurements is posed in DIP as 
\begin{equation}\label{dip}
\{\boldsymbol \theta^*\} = \arg \min_{\boldsymbol\theta}\left\|\mathcal A(\mathbf x) - \mathbf b\right\|^2 ~~\mbox{such that} ~~ \mathbf x = \mathcal G_{\boldsymbol \theta}[\mathbf z]
\end{equation} 
where $\mathbf x=\mathcal G_{\boldsymbol \theta^*}(\mathbf z)$ is the recovered image, generated by the CNN generator $\mathcal G_{\boldsymbol\theta^*}$ whose parameters are denoted by $\boldsymbol\theta$. Here, $\mathbf z$ is the random latent variable, which is chosen as random noise and kept fixed. 

The above optimization problem is often solved using stochastic gradient descent (SGD). Since CNNs are efficient in learning natural images, the solution often converges quickly to a good image. However, when iterated further, the algorithm also learns to represent the noise in the measurements if the generator has sufficient capacity, resulting in poor image quality. The general practice is to rely on early termination to obtain good results. This approach was recently extended to the dynamic setting by Jin et al. \cite{jin2019time}, where a sequence of random vectors was used as the input.

\section{Deep generative SToRM model}
We now introduce a synthesis SToRM formulation for the recovery of images in a time series from undersampled data (see Fig. \ref{illus}.(b)). Rather than relying on a non-linear mapping of images to a low-dimensional subspace  \cite{poddar2019manifold} (see Fig. \ref{illus}.(a)), we model the images in the time series as non-linear functions of latent vectors living in a low-dimensional subspace. 

\subsection{Generative model}
We model the images in the time series as 
\begin{equation}\label{genmodel}
\mathbf x_i = \mathcal G_{\theta}(\mathbf z_i), i=1,..,M,
\end{equation}
where $\mathcal G_{\theta}$ is a non-linear mapping, which is termed as the generator. Inspired by the extensive work on generative image models \cite{goodfellow2014generative,arjovsky2017wasserstein,ulyanov2018deep}, we represent $\mathcal G_{\theta}$ by a deep CNN, whose weights are denoted by $\theta$. The parameters $\mathbf z_i$ are the latent vectors, which live in a low-dimensional subspace. The non-linear mapping $\mathcal G_{\theta}$ may be viewed as the inverse of the image-to-latent space mapping $\varphi$, considered in the SToRM approach. 
\begin{figure}[!h]
	\begin{center}
\includegraphics[width=\textwidth]{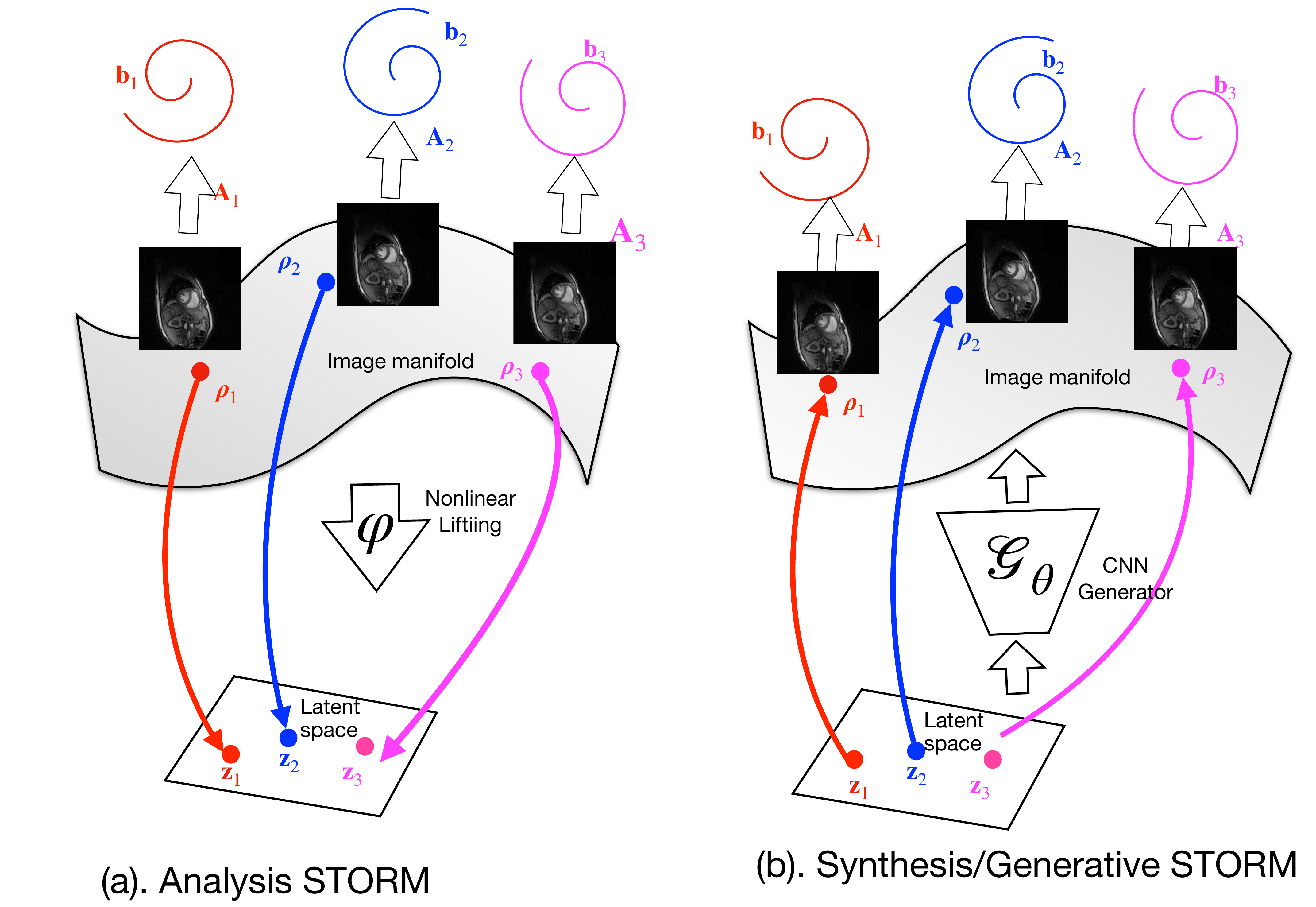}
	\end{center}
\caption{Illustration of (a) analysis SToRM and (b) generative SToRM. Analysis SToRM considers a non-linear (e.g. exponential) lifting of the data. If the original points lie on a smooth manifold, the lifted points lie on a low-dimensional subspace. The analysis SToRM cost function in \eqref{storm} is the sum of the fit of the recovered images to the undersampled measurements and the nuclear norm of the lifted points. A challenge with analysis SToRM is its high memory demand and the difficulty in adding spatial regularization. The proposed method models the images as the non-linear mapping $\mathcal G_{\theta}$ of some latent vectors $\mathbf{z}_i$, which lie in a very low-dimensional space. Note that the same generator is used to model all the images in the dataset. The number of parameters of the generator and the latent variables is around the size of a single image, which implies a highly compressed representation. In addition, the structure of the CNN offers spatial regularization as shown in DIP. The proposed algorithm in \eqref{gen_storm} estimates the parameters of the generator and the latent variables from the measured data. A distance regularization prior is added to the generator to ensure that nearby points in the latent subspace are mapped to nearby points on the manifold. Similarly, a temporal regularization prior is added to the latent variables. The optimization is performed using ADAM with batches of few images. }
\label{illus}
\end{figure}

We propose to estimate the parameters of the network $\theta$ as well as the latent vectors $\mathbf z_i$ by fitting the model to the undersampled measurements. The main distinction of our framework with DIP, which is designed for a single image, is that we use the same generator for all the images in the dynamic dataset. The latent vector $\mathbf z_i$ for each image is different and is also estimated from the measurements. This strategy allows us to exploit non-local information in the dataset. For example, in free-breathing cardiac MRI, the latent vectors of images with the same cardiac and respiratory phase are expected to be similar. When the gradient of the network is bounded, the output images at these time points are expected to be the same. The proposed framework is hence expected to learn a common representation from these time-points, which are often sampled using different sampling trajectories. Unlike conventional manifold methods \cite{poddar2015dynamic,poddar2019manifold,nakarmi2018mls}, the use of the CNN generator also offers spatial regularization.

It is often impossible to acquire fully-sampled training data in many free-breathing dynamic imaging applications, and a key benefit of this framework over conventional neural network schemes  is that no training data is required. As discussed previously, the number of parameters of the model in \eqref{genmodel} is orders of magnitude smaller than the number of pixels in the dataset. The dramatic compression offered by the representation, together with the mini-batch training provides a highly memory-efficient alternative to current manifold based and low-rank/tensor approaches. Although our focus is on establishing the utility of the scheme in 2-D settings in this paper, the approach can be readily translated to higher dimensional applications. Another benefit is the implicit spatial regularization brought in by the convolutional network as discussed for DIP. We now introduce novel regularization priors on the network and the latent vectors to further constrain the recovery to reduce the need for manual early stopping. 

\subsection{Distance/Network regularization}
As in the case of analysis SToRM regularization \cite{poddar2015dynamic,poddar2019manifold}, our interest is in generating a manifold model that preserves distances. Specifically, we would like the nearby points in the latent space to map to similar images on the manifold. With this interest, we now study the relation between the Euclidean distances between their latent vectors and the shortest distance between the points on the manifold (geodesic distance).

We consider two points $\mathbf z_1$ and $\mathbf z_2$ in the latent space, which are fed to the generator to obtain $\mathcal G(\mathbf z_1)$ and $\mathcal G(\mathbf z_2)$, respectively. We have the following result, which relates the the Euclidean distance $\|\mathbf z_1-\mathbf z_2\|^2$ to the geodesic distance ${\rm dist}_{\mathcal M}\left(\mathcal G(\mathbf z_1),\mathcal G(\mathbf z_2)\right)$, which is the shortest distance on the manifold. The setting is illustrated in Fig. \ref{illusdistance}, where the geodesic distance is indicated by the red curve.
\begin{prop}
	Let $\mathbf z_1,\mathbf z_2 \in \mathbb R^n$ be two nearby points in the latent space, with mappings denoted by $\mathcal G(\mathbf z_1),\mathcal G(\mathbf z_2) \in \mathcal M$. Here, $\mathcal M = \{G(\mathbf z)|\mathbf z \in \mathbb R^n\}$. Then, the geodesic distance on the manifold satisfies:
	\begin{equation}\label{thmresult}
	{\rm dist}_{\mathcal M}\big(\mathcal G(\mathbf z_1),\mathcal G(\mathbf z_2)\big) \leq \|\mathbf z_1-\mathbf z_2\|_F ~\| J_z\big(\mathcal G\left(\mathbf z_1\right)\big) \|_F .
	\end{equation}
\end{prop}
\begin{proof}
The straight-line between the latent vectors is denoted by $c(s), s\in[0,1]$ with $c(0)=\mathbf z_1$ and $c(1)=\mathbf z_2$. We also assume that the line is described in its curvilinear abscissa, which implies $\|c'(s)\|=1; \forall s \in [0,1]$. We note that $\mathcal G$ may map to the black curve, which may be longer than the geodesic distance. We now compute the length of the black curve $\mathcal G[c(s)]$ as 
\begin{equation}\label{distance}
d = \int_{0}^{1} \|\nabla_s \mathcal G\left[c(s)\right]\|ds.
\end{equation}
Using the chain rule and denoting the Jacobian matrix of $\mathcal G$ by $J_z$, we can simplify the above distance as
\begin{eqnarray}\nonumber
d &=& \int_{0}^{1} \|J_z \left(\mathcal G\right) c'(s) \|_F ds\\\nonumber
&\leq & \int_{0}^{1} \|J_z \left(\mathcal G\right) \|_F~ \underbrace{\|c'(s)\|_F}_{1} ds\\
&=&\|J_z \left(\mathcal G[\mathbf z_1]\right) \|_F \underbrace{\int_{0}^{1}ds}_{\|\mathbf z_1-\mathbf z_2\|}.
\end{eqnarray}
We used the Cauchy-Schwartz inequality in the second step and in the last step, we use the fact that $J_z \mathcal G\left(c(t)\right) = J_z \mathcal G\left(\mathbf z_1\right)+ \mathcal O(t)$ when the points $\mathbf z_1$ and $\mathbf z_2$ are close. Since the geodesic distance is the shortest distance on the manifold, we have ${\rm dist}_{\mathcal M}\big(\mathcal G(\mathbf z_1),\mathcal G(\mathbf z_2)\big) \leq d$ and hence we obtain \eqref{thmresult}.
\end{proof}
 
 The result in \eqref{thmresult} shows that the Frobenius norm of the Jacobian matrix $\|J_z \mathcal G\|$ controls how far apart $\mathcal G$ maps two vectors that are close in the latent space. We would like points that are close in the latent space map to nearby points on the manifold. We hence use the gradient of the map:
\begin{equation}\label{key}
R_{\rm distance}  = \|J_z\big( \mathcal G(\mathbf z) \big)\|_F^2
\end{equation}
as a regularization penalty. We note that the above penalty will also encourage the learning of a mapping $\mathcal G$ such that the length of curve $\mathcal G(c(t))$ is the geodesic distance. We note that the above penalty can also be thought of as a network regularization. Similar gradient penalties are used in machine learning to improve generalization ability and to improve the robustness to adversarial attacks \cite{varga2017gradient}. 
The use of gradient penalty is observed to be qualitatively equivalent to penalizing the norm of the weights of the network. 

\begin{figure}[!h]
	\begin{center}
		\includegraphics[width=0.7\textwidth]{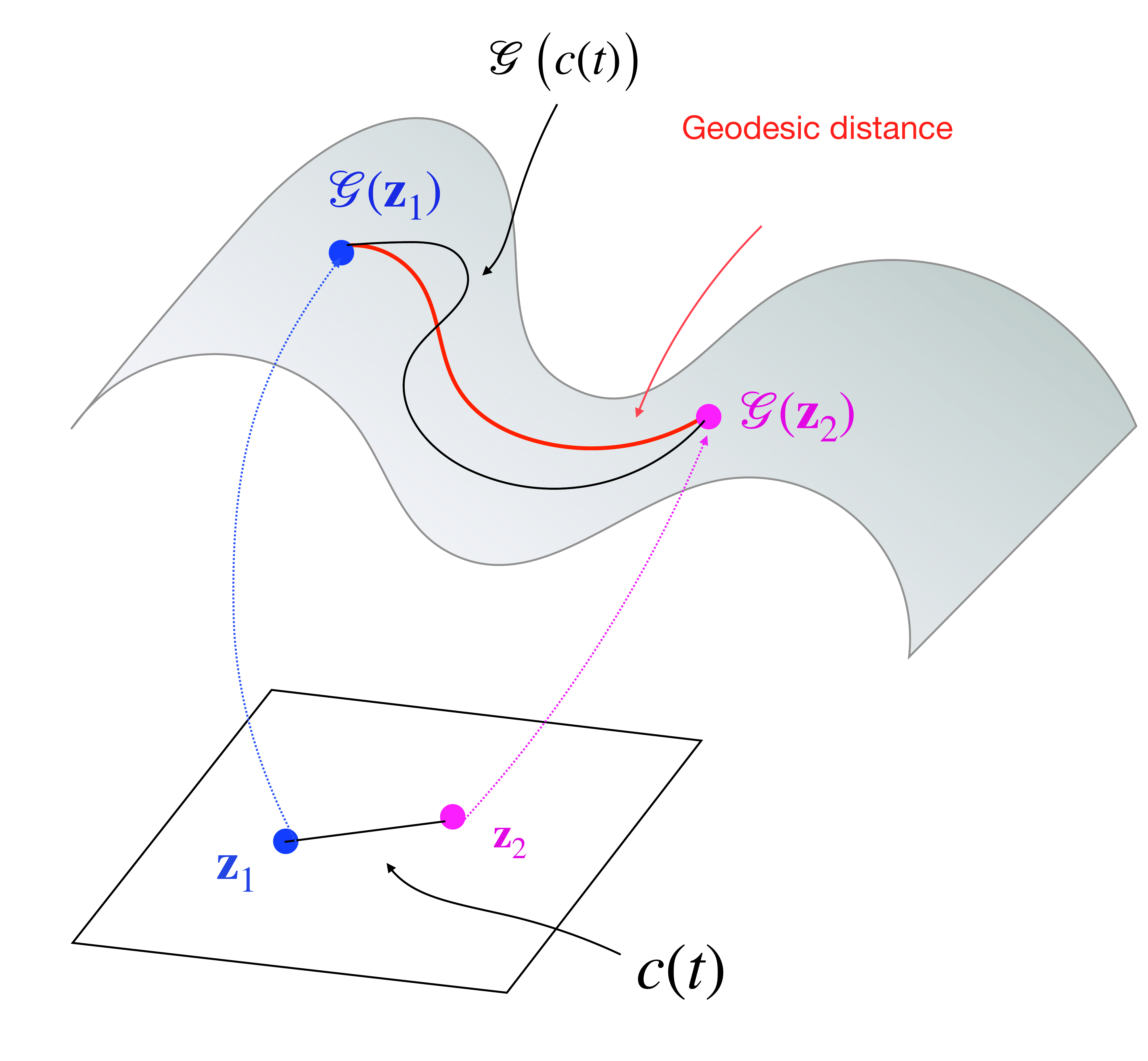}
	\end{center}
	\caption{Illustration of the distance penalty. The length of the curve connecting the images corresponding to $\mathbf z_1$ and $\mathbf z_2$ depends on the Frobenius norm of the Jacobian of the mapping $\mathcal G$ as well as the Euclidean distance $\|\mathbf z_1-\mathbf z_2\|^2$. We are interested in learning a mapping that preserves distances; we would like nearby points in the latent space to map to similar images. We hence use the norm of the Jacobian as the regularization prior, with the goal of preserving distances.  }
\vspace{-1\baselineskip}
	\label{illusdistance}
\end{figure}

\subsection{Latent vector regularization penalty}
The time frames in a dynamic time series have extensive redundancy between adjacent frames, which is usually capitalized by temporal gradient regularization. Directly penalizing the temporal gradient norm of the images requires the computation of the entire image time series, which is difficult when the entire image time series is not optimized in every batch. 

We consider the norm of the finite differences between images specified by $\|\nabla_p\mathbf G[\mathbf z_p] \|^2$. Using Taylor series expansion, we obtain $\nabla_p\mathbf G[\mathbf z_p] = J_{\mathbf z}(\mathcal G[\mathbf z]) \nabla_p \mathbf{z} + \mathcal O(p)$. We thus have
\begin{equation}\label{key}
\|\nabla_p\mathcal G[\mathbf z_p] \| \approx \| J_{\mathbf z}(\mathcal G[\mathbf z]) \nabla_p \mathbf{z}\| \leq \|J_{\mathbf z}(\mathcal G[\mathbf z])\|~ \|\nabla_{p} \mathbf z\|.
\end{equation}
Since $ J_{\mathbf z}(\mathcal G[\mathbf z])$ is small because of the distance regularization, we propose to add a temporal regularizer on the latent vectors. For example, when applied to free-breathing cardiac MRI, we expect the latent vectors to capture the two main contributors of motion: cardiac motion and respiratory motion. The temporal regularization encourages the cardiac and respiratory phases change slowly in time. 

\subsection{Proposed optimization criterion}
Based on the above analysis, we derive the parameters of the network $\theta$ and the low-dimensional latent vectors $\mathbf z_i; i=1,..,M$ from the measured data by minimizing:
\begin{eqnarray}\nonumber\label{gen_storm}
\mathcal C(\mathbf z,\theta)&=&  \underbrace{\sum_{i=1}^N\|\mathcal A_i\left(\mathcal G_{\theta}[\mathbf z_i]\right) - \mathbf b\|^2}_{\scriptsize\mbox{data term}} + \lambda_1 \underbrace{\|J_{\mathbf z} \mathcal G_{\theta}(\mathbf{z})\|^2}_{\scriptsize \mbox{distance regularization}}\\\label{gen-storm}&&\qquad +  \lambda_2 \underbrace{\|\nabla_{t} \mathbf z_t\|^2 }_{\scriptsize\mbox{latent regularization}}
\end{eqnarray}
with respect to $\mathbf z$ and $\theta$. We use the ADAM optimization to determine the optimal parameters, and random initialization is used for the network parameters and latent variables. 

A potential challenge with directly solving \eqref{gen-storm} is its high computational complexity. Unlike supervised neural network approaches that offer fast inference, the proposed approach optimizes the network parameters based on the measured data. This strategy will amount to a long reconstruction time when there are several image frames in the time series. 

\subsection{Strategies to reduce computational complexity}
To  minimize the computational complexity, we now introduce some approximation strategies.

\subsubsection{Approximate data term for accelerated convergence}\label{dataterm}
When the data is measured using non-Cartesian sampling schemes, $M$ non-uniform fast Fourier transform (NUFFT) evaluations are needed for the evaluation of the data term, where $M$ is the number of frames in the dataset. Similarly, $M$ inverse non-uniform fast Fourier transform (INUFFT) evaluations are needed for each back-propagation step. These NUFFT evaluations are computationally expensive, resulting in slow algorithms. In addition, most non-Cartesian imaging schemes over-sample the center of k-space. Since the least-square loss function in \eqref{storm} weights errors in the center of k-space higher than in outer k-space regions, it is associated with slow convergence.

To speed up the intermediate computations, we propose to use gridding with density compensation, together with a projection step for the initial iterations. Specifically, we will use the approximate data term 
\begin{equation}\label{proj}
D(\mathbf z,\theta) = \sum_{i=1}^M \|\mathcal P_i\left(\mathcal G_{\theta}[\mathbf z_i]\right) -\mathbf g_i  \|^2
\end{equation}
instead of $\sum_i\|\mathcal A_i\left(\mathcal G[\mathbf z_i]\right)-\mathbf b_i\|^2$ in early iterations to speed up the computations. Here, $\mathbf g_i$ are the gridding reconstructions
\begin{equation}\label{key}
\mathbf g_i =  \left(\mathcal A_i^H\mathcal A_i\right)^{\dag} \mathcal A_i^H  ~\mathbf b_i \approx \mathcal A_i^H~ \mathcal W~\mathbf b,
\end{equation}
where, $\mathcal W$ are diagonal matrices corresponding to multiplication by density compensation factors. The operators $\mathcal P_i$ in \eqref{proj} are projection operators:
\begin{equation}
\mathcal P_i ~\mathbf x = \left(\mathcal A_i^H\mathcal A_i\right)^{\dag} \left(\mathcal A_i^H \mathcal A_i\right) ~\mathbf x \approx \left(\mathcal A_i^H~\mathcal W~\mathcal A_i\right) \mathbf x
\end{equation}
 We note that the term $\left(\mathcal A_i^H~\mathcal W~\mathcal A_i\right) \mathbf x$ can be efficiently computed using Toeplitz embedding, which eliminates the need for expensive NUFFT and INUFFT steps. In addition, the use of the density compensation serves as a preconditioner, resulting in faster convergence. Once the algorithm has approximately converged, we switch the loss term back to \eqref{storm} since it is optimal in a maximum likelihood perspective. 

\subsubsection{Progressive training-in-time}\label{ptt}
To further speed up the algorithm, we introduce a progressive training strategy, which is similar to multi-resolution strategies used in image processing. In particular, we start with a single frame obtained by pooling the measured data from all the time frames. Since this average frame is well-sampled, the algorithm promptly converges to the optimal solution. The corresponding network serves as a good initialization for the next step. Following convergence, we increase the number of frames. The optimal $\theta$ parameters from the previous step are used to initialize the generator, while the latent vector is initialized by the interpolated version of the latent vector at the previous step. This process is repeated until the desired number of frames is reached. 

\begin{figure*}[!h]
	\begin{center}
\includegraphics[width=0.8\textwidth]{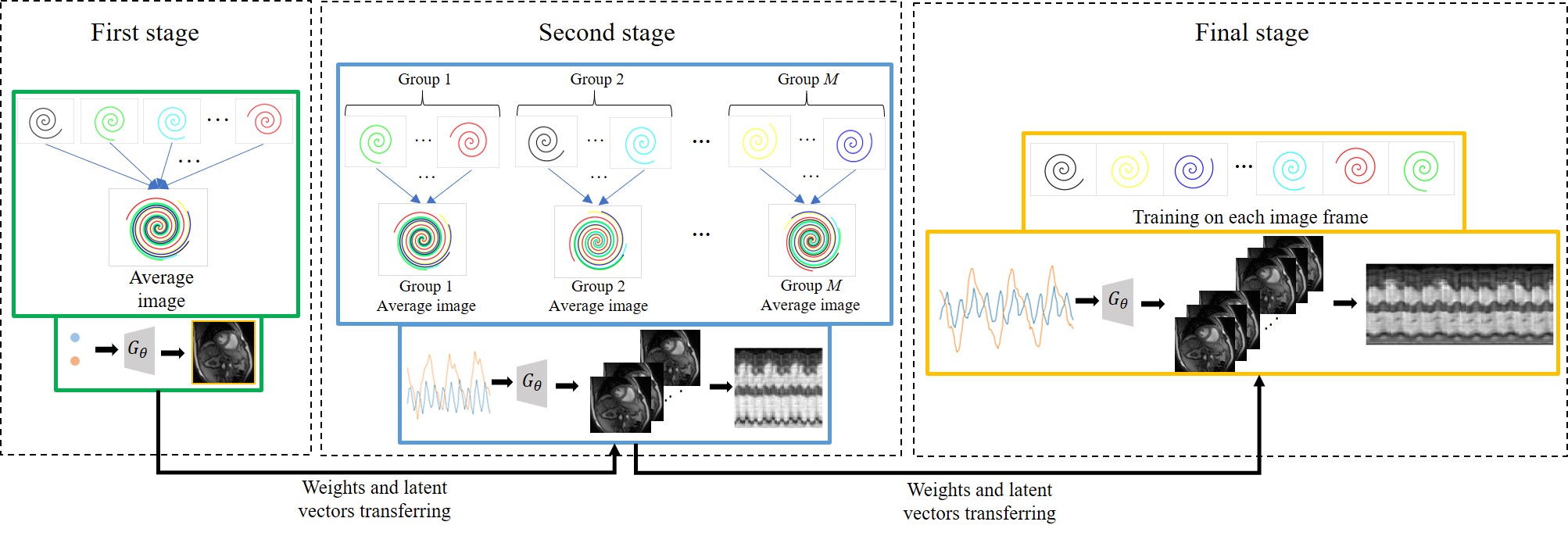}
\caption{Illustration of the progressive training-in-time approach. In the first level of training, the k-space data of all the frames are binned into one and we try to solve for the average image in this level. Upon the convergence of the first step, the parameters and latent variables are transferred as the initialization of the second step. In the second level of training, we divide the k-space data into $M$ groups and try to reconstruct the $M$ average images. Following the convergence, we can move to the final level of training, where the parameters obtained in the second step and the linear interpolation of the latent vectors in the second step are chosen as the initializations of the final step of training.}
\label{progressive}
\end{center}
\end{figure*}

This progressive training-in-time approach significantly reduces the computational complexity of the proposed algorithm. In this work, we used a three-step algorithm. However, the number of steps (levels) of training can be chosen based on the dataset. This progressive training-in-time approach is illustrated in Fig. \ref{progressive}.

\section{Implementation details and datasets}

\subsection{Structure of the generator}
The structure of the generator used in this work is given in Table. \ref{generator}. The output images have two channels, which correspond to the real and imaginary parts of the MR images. Note that we have a parameter $d$ in the network. This user-defined parameter controls the size of the generator or, in other words, the number of trainable parameters in the generator. We also have a number $\ell({\mathbf{z}})$ as a user-defined parameter. This parameter represents the number of elements in each latent vector. In this work, it is chosen as $\ell({\mathbf{z}})=2$ as we have two motion patterns in cardiac images. We use leaky ReLU for all the non-linear activations, except at the output layer, where it is tanh activation.
\begin{table}[!h]
	\centering
	\begin{tabular}{|c | c | c| c| c|c|c| }
		\hline\scriptsize
		Input size & \scriptsize filter sz &\scriptsize \# filters & \scriptsize Padding & \scriptsize Stride & \scriptsize Output size \\
		\hline \hline\scriptsize
		$1\times 1\times \ell({\mathbf{z}})$ &\scriptsize $1\times 1$ &\scriptsize 100 &\scriptsize 0 &\scriptsize 1 &\scriptsize $1\times 1\times 100$  \\ 
		\hline\scriptsize
		$1\times 1\times 100$ &\scriptsize $3\times 3$ &\scriptsize $8d$ &\scriptsize 0 &\scriptsize 1 &\scriptsize $3\times 3\times 8d$  \\ 
		\hline\scriptsize
		$3\times 3\times 8d$ &\scriptsize $3\times 3$ &\scriptsize $8d$ &\scriptsize 0 &\scriptsize 1 &\scriptsize $5\times 5\times 8d$ \\ 
		\hline\scriptsize
		$5\times 5\times 8d$ &\scriptsize $4\times 4$ &\scriptsize $4d$ &\scriptsize  1 &\scriptsize  2 &\scriptsize $10\times 10\times 4d$  \\ 
		\hline\scriptsize
		$10\times 10\times 4d$ &\scriptsize $4\times 4$ &\scriptsize $4d$ &\scriptsize  1 &\scriptsize  2 &\scriptsize $20\times 20\times 4d$  \\ 
		\hline\scriptsize
		$20\times 20\times 4d$ &\scriptsize $3\times 3$ &\scriptsize $4d$ &\scriptsize  0 &\scriptsize  2 &\scriptsize $41\times 41\times 4d$  \\ 
		\hline\scriptsize
		$41\times 41\times 4d$ &\scriptsize $5\times 5$ &\scriptsize $2d$ &\scriptsize  1 &\scriptsize  2 &\scriptsize $85\times 85\times 2d$  \\ 
		\hline\scriptsize
		$85\times 85\times 2d$ &\scriptsize $4\times 4$ &\scriptsize $d$ &\scriptsize  1 &\scriptsize  2 &\scriptsize $170\times 170\times d$  \\ 
		\hline\scriptsize
		$170\times 170\times d$ &\scriptsize $4\times 4$ &\scriptsize $d$ &\scriptsize  1 &\scriptsize  2 &\scriptsize $340\times 340\times d$  \\ 
		\hline\scriptsize
		$340\times 340\times d$ &\scriptsize $3\times 3$ &\scriptsize $2$ &\scriptsize  1 &\scriptsize  2 &\scriptsize $340\times 340\times 2$  \\ 
		\hline
	\end{tabular} 
	\caption{Architecture of the generator $\mathcal G_{\theta}$. $\ell(\mathbf{z})$ means the number of elements in each latent vector.}
\vspace{-1\baselineskip}
	\label{generator}
\end{table}

\subsection{Datasets}

This research study was conducted using data acquired from human subjects. The Institutional Review Board at the local institution (The University of Iowa) approved the acquisition of the data, and written consents were obtained from all subjects. The experiments reported in this paper are based on datasets collected in the free-breathing mode using the golden angle spiral trajectory.We acquired eight datasets on a GE 3T scanner. One dataset was used to identify the optimal hyperparameters of all the algorithms in the proposed scheme. We then used the hyperparameters to generate the experimental results for all the remaining datasets reported in this paper. The sequence parameters for the datasets are: TR = 8.4 ms, FOV = 320 mm$\times$ 320 mm, flip angle = 18$^\circ$, slice thickness = 8 mm. The datasets were acquired using a cardiac multichannel array with 34 channels. We used an automatic algorithm to pre-select the eight best coils, that provide the best signal to noise ratio in the region of interest. The removal of the coils with low sensitivities provided improved reconstructions \cite{zhou2019free}. We used a PCA-based coil combination using SVD such that the approximation error  $<$ 5\%. We then estimated the coil sensitivity maps based on these virtual channels  using the method of Walsh et al. \cite{walsh2000adaptive} and assumed they were constant over time.

For each dataset in this research, we binned the data from six spiral interleaves corresponding to 50 ms temporal resolution. If a Cartesian acquisition scheme with $TR=3.5ms$ were used, this would correspond to $\approx$14 lines/frame; with a $340\times 340$ matrix, this corresponds roughly to an acceleration factor of 24. Moreover, each dataset has more than 500 frames. During reconstruction, we omit the first 20 frames in each dataset and use the next 500 frames for SToRM reconstructions; this is then used as the simulated ground truth for comparisons. The experiments were run on a machine with an Intel Xeon CPU at 2.40 GHz and a Tesla P100-PCIE 16GB GPU. The source code for the proposed Gen-SToRM scheme can be downloaded from this link: {\url{https://github.com/qing-zou/Gen-SToRM}}.

\subsection{Quality evaluation metric}

In this work, the quantitative comparisons are made using the Signal-to-Error Ratio (SER) metric (in addition to the standard Peak Signal-to-Noise Ratio (PSNR) and the Structural Similarity Index Measure (SSIM)) defined as:
\[{\mathrm{SER}} = 20\cdot\log_{10}\frac{\|\mathbf{x}_{orig}\|}{\|\mathbf{x}_{orig}-\mathbf{x}_{recon}\|}.\]
Here $\mathbf{x}_{orig}$ and $\mathbf{x}_{recon}$ represent the ground truth and the reconstructed image. The unit for SER is decibel (dB). 

The SER metric requires a reference image, which is chosen as the SToRM reconstruction with 500 frames. However, we note that this reference may be imperfect and may suffer from blurring and related artifacts. Hence, we consider the Blind/referenceless Image Spatial Quality Evaluator (BRISQUE) \cite{mittal2012no} to evaluate the score of the image quality. The BRISQUE score is a perceptual score based on the support vector regression model trained on an image database with corresponding differential mean opinion score values. The training image dataset contains images with different distortions. A smaller score indicates better perceptual quality.

\begin{figure}[!b]
	\centering
	\subfigure[Performance comparison]{\includegraphics[width=0.37\textwidth]{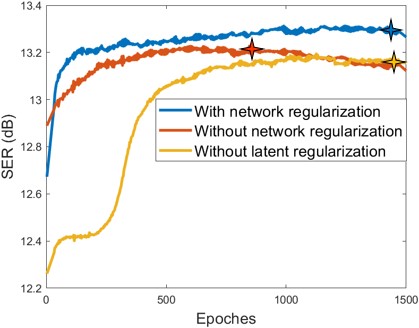}}
	\subfigure[Latent codes with both terms]{\includegraphics[width=0.47\textwidth]{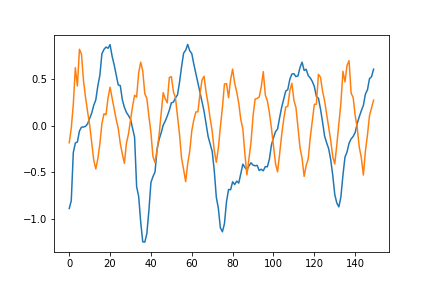}}\\
	\subfigure[Without distance regularization]{\includegraphics[width=0.47\textwidth]{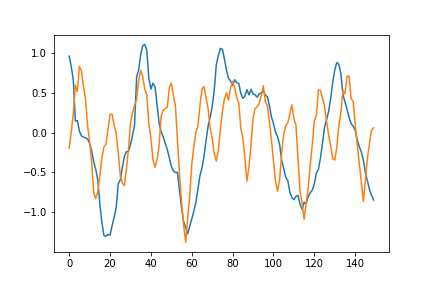}}
	\subfigure[Without latent regularization]{\includegraphics[width=0.47\textwidth]{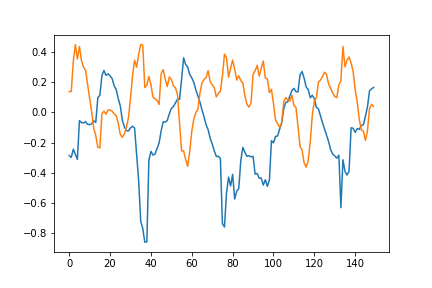}}\\
	\subfigure[Visual and quantitative comparisons]{\includegraphics[width=0.9\textwidth]{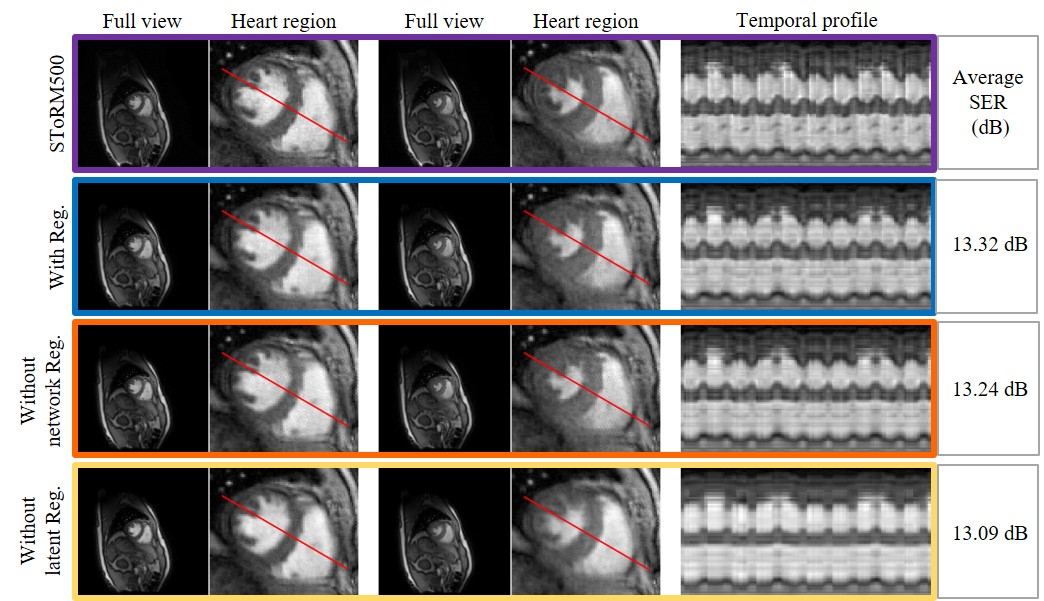}}
	\caption{Illustration of the impact of the regularization terms in the proposed scheme with $d=24$. We considered three cases in the experiment: (1) using both regularizations, (2) using only latent regularization, and (3) using only network regularization; these correspond to the blue, orange, and yellow curves in (a). In (b), (c), and (d), we showed the learned latent vectors for the three cases. The visual and quantitative comparisons of the three cases are shown in (e).}
	\label{reg_impact}
\end{figure}

\subsection{State-of-the-art methods for comparison}\label{meths}

We compare the proposed scheme with the recent state-of-the-art methods for free-breathing and ungated cardiac MRI. We note that while there are many deep learning algorithms for static MRI, those methods are not readily applicable to our setting.
\begin{itemize}
\item Analysis SToRM \cite{poddar2019manifold,ahmed2020free}, published in 2020: The manifold Laplacian matrix is estimated from k-space navigators using kernel low-rank regularization, followed by solving for the images using \eqref{storm1}.
\item Time-DIP \cite{jin2019time} implementation based on the arXiv form at the submission of this article: This is an unsupervised learning scheme, that fixes the latent variables as noise and solves for the generator parameters. For real-time applications, Time-DIP chooses a preset period, and the noise vectors of the frames corresponding to the multiples of the period were chosen as independent Gaussian variables \cite{jin2019time}. The latent variables of the intermediate frames were obtained using linear interpolation. We chose a period of 20 frames, which roughly corresponds to the period of the heart beats.
\item Low-rank \cite{lingala2011accelerated}: The image frames in the time series are recovered using the nuclear norm minimization. 
\end{itemize}

\subsection{Hyperparameter tuning}

We used one of the acquired datasets to identify the hyperparameters of the proposed scheme. Since we do not have access to the fully-sampled dataset, we used the SToRM reconstructions from 500 images (acquisition time of $25$ seconds) as a reference. The smoothness parameter $\lambda$ of this method was manually selected as $\lambda=0.01$ to obtain the best recovery, as in the literature \cite{ahmed2020free}. All of the comparisons relied on image recovery from 150 frames (acquisition time of 7.5 seconds). The hyperparameter tuning approach yielded the parameters $d=40$, $\lambda_1 = 0.0005$, and $\lambda_2 =2$ for the proposed approach. We demonstrate the impact of tuning $d$ in Fig. \ref{generator_size}, while the impact of choosing $\lambda_1$ and $\lambda_2$ is shown in Fig. \ref{reg_impact}. The hyperparameter optimization of SToRM from 150 frames resulted in the optimal smoothness parameter $\lambda=0.0075$. For Time-DIP, we follow the design of the network shown by Jin et al. \cite{jin2019time}, where the generator consists of multiple layers of convolution and upsampling operations. To ensure fair comparison, we used a  similar architecture, where the base size of the network was tuned to obtain the best results.

We use a three-step progressive training strategy. In the first step, the learning rate for the network is $1\times 10^{-3}$ and 1000 epoches are used. For the second step of training, the learning rate for the network is $5\times 10^{-4}$ and the learning rate for the latent variable is $5\times 10^{-3}$. In this stage, 600 epoches are used. In the final step of training, the learning rate for the network is $5\times 10^{-4}$, the learning rate for the latent variable is $1\times 10^{-3}$, and 700 epoches are used.

\section{Experiments and results}

\subsection{Impact of different regularization terms}
We first study the impact of the two regularization terms in \eqref{gen_storm}. The parameter $d$ corresponding to the size of the network (see Table \ref{generator}) was chosen as $d=24$ in this case. In Fig. \ref{reg_impact} (a), we plot the reconstruction performance with respect to the number of epoches for three scenarios: (1) using both regularization terms; (2) using only latent regularization; and (3) using only distance/network regularization. In the experiment, we use 500 frames of SToRM ($\sim$ 25 seconds of acquisition) reconstructions, which is called ``SToRM500'', as the reference for SER computations. We tested the reconstruction performance for the three scenarios using 150 frames, which corresponds to around 7.5 seconds of acquisition. From the plot, we observe that without using the network regularization, the SER degrades with increasing epoches, which is similar to that of DIP. In this case, an early stopping strategy is needed to obtain good recovery. The latent vectors corresponding to this setting are shown in (c), which shows mixing between cardiac and respiratory waveforms. When latent regularization is not used, we observe that the SER plot is roughly flat, but the latent variables show quite significant mixing, which translates to blurred reconstructions. By contrast, when both network and latent regularizations are used, the algorithm converges to a better solution. We also note that the latent variables are well decoupled; the blue curve captures the respiratory motion, while the orange one captures the cardiac motion. We also observe that the reconstructions agree well with the SToRM reconstructions. The network now learns meaningful mappings, which translate to improved reconstructions when compared to the reconstructions obtained without using the regularizers. 

\begin{figure}[!h]
	\centering
	\includegraphics[width=0.7\textwidth]{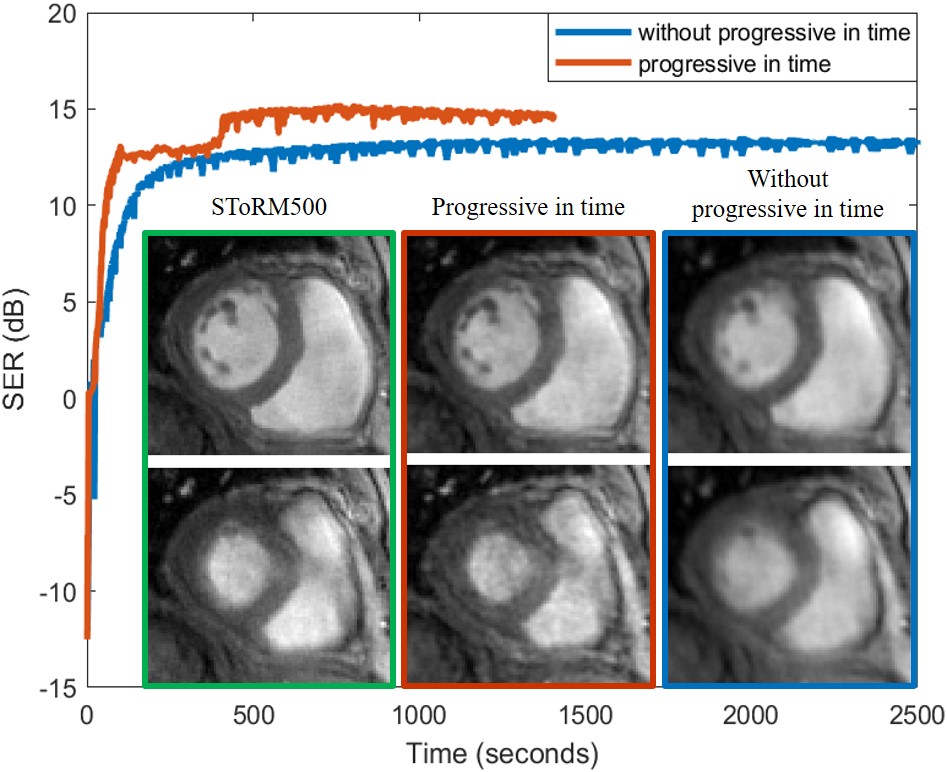}
	\caption{Comparisons of the reconstruction performance with and without the progressive training-in-time strategy using $d=40$. From the plot of SER vs. running time, we can see that the progressive training-in-time approach yields better results with much less running time comparing to the training without using progressive training-in-time. Two reconstructed frames near the end of systole and diastole using SToRM500, the proposed scheme with progressive training-in-time and the proposed scheme without using the progressive training-in-time are shown in the plot as well for comparison purposes. The average Brisque scores for SToRM500, the reconstruction with progressive training-in-time, and the reconstruction without progressive training-in-time are $36.4, 37.3$ and $39.1$ respectively.}
	\label{PGvsNPG}
\end{figure}

\subsection{Benefit of progressive training-in-time approach}


In Fig. \ref{PGvsNPG}, we demonstrate the significant reduction in run-time offered by the progressive training strategy described in Section \ref{ptt}. Here, we consider the recovery from 150 frames with and without the progressive strategy. Both regularization priors were used in this strategy, and $d$ was chosen as 24. We plot the reconstruction performance, measured by the SER with respect to the running time. The SER plots show that the proposed scheme converges in around $\approx 200$ seconds, while the direct approach takes more than 2000 seconds. We also note from the SER plots that the solution obtained using progressive training is superior to the one without progressive training.

\subsection{Impact of size of the network}

The architecture of the generator $\mathcal{G}_{\theta}$ is given in Table \ref{generator}. Note that the size of the network is controlled by the user-defined parameter $d$, which dictates the number of convolution filters and hence the number of trainable parameters in the network. In this section, we investigate the impact of the user-defined parameter $d$ on the reconstruction performance. We tested the reconstruction performance using $d = 8, 16, 24, 32, 40$, and $48$, and the obtained results are shown in Fig. \ref{generator_size}. From the figure, we see that when $d = 8$ or $d=16$, the generator network is too small to capture the dynamic variations. When $d = 8$, the generator is unable to capture both cardiac motion and respiratory motion. When $d=16$, part of the respiratory motion is recovered, while the cardiac motion is still lost. The best SER scores with respect to SToRM with 500 frames is obtained for $d=24$, while the lowest Brisque scores are obtained for $d=40$. We also observe that the features including papillary muscles and myocardium in the $d=40$ results appear sharper than those of SToRM with 500 frames, even though the proposed reconstructions were only performed from 150 frames. We use $d=40$ for the subsequent comparisons in the paper.

\begin{figure}[!h]
	\centering
	\includegraphics[width=0.9\textwidth]{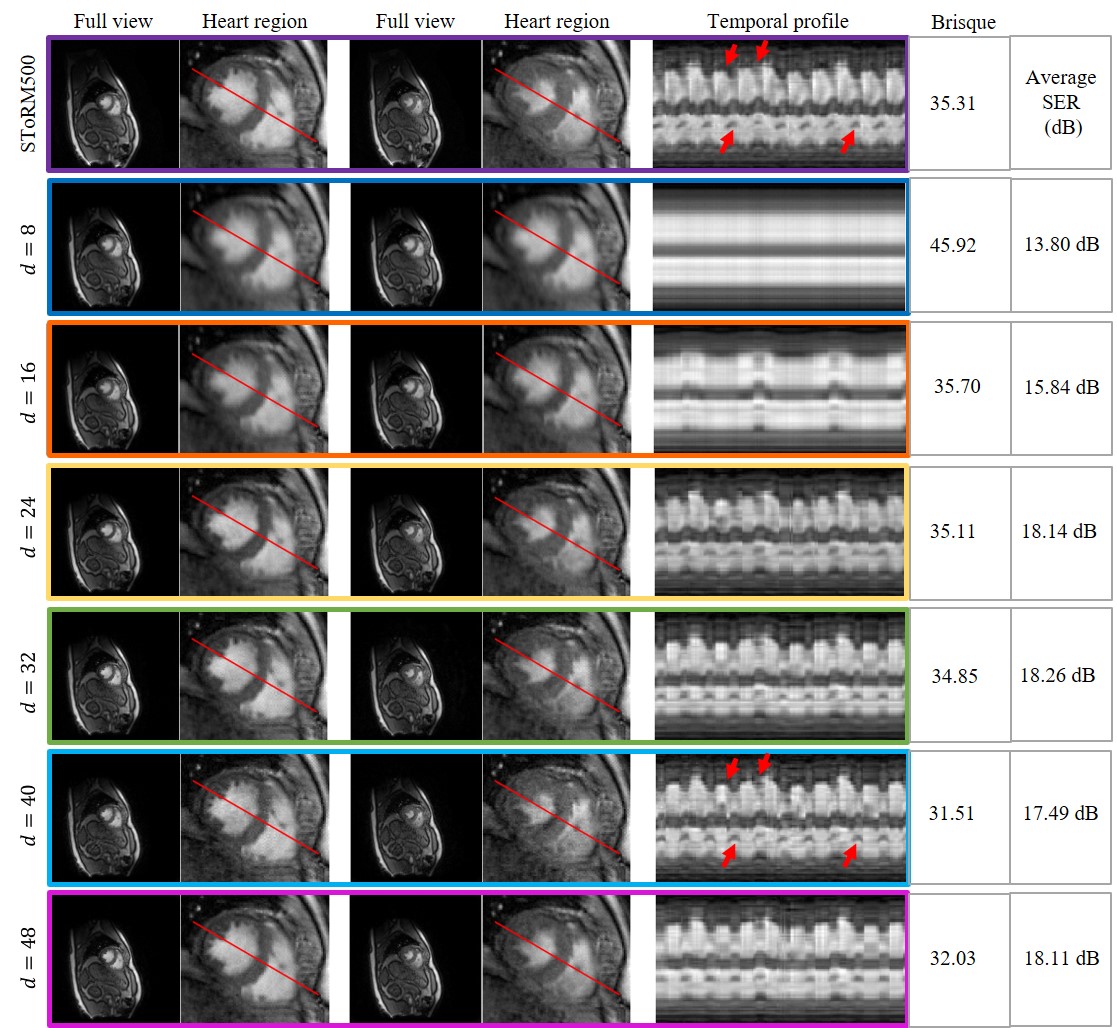}
	\caption{Impact of network size on reconstruction performance. In the experiments, we chose $d = 8, 16, 24, 32, 40$ and $48$ to investigate the reconstruction performance. We used 500 frames for SToRM reconstructions (SToRM500) as the reference for SER comparisons. For the investigation of the impact of network size on the reconstructions, we used 150 frames. The diastolic and systolic states and the temporal profiles are shown in the figure for each case. The Brisque scores and average SER are also reported. It is worth noting that when $d=40$, the results are even less blurred than the SToRM500 results, even though only one-third of the data are used.}
	\label{generator_size}
\end{figure}

\subsection{Comparison with the state-of-the-art methods}

In this section, we compare our proposed scheme with several state-of-the-art methods for the reconstruction of dynamic images. 

\begin{table}[!h]
	\centering
	\begin{tabular}{|c | c | c| c|c| }
		\hline
		\scriptsize Methods &\scriptsize SToRM500 &\scriptsize SToRM150  &\scriptsize  Propsed &\scriptsize  Time-DIP   \\
		\hline \hline
		\scriptsize SER (dB) &\scriptsize NA &\scriptsize$17.3$ &\scriptsize $\mathbf{18.2}$ \scriptsize& \scriptsize $16.7$  \\ 
		\hline
		\scriptsize PSNR (dB) &\scriptsize NA &\scriptsize $32.7$ &\scriptsize  $\mathbf{33.5}$ &\scriptsize $32.0$    \\ 
		\hline
		\scriptsize SSIM &\scriptsize NA &\scriptsize $0.86$ &\scriptsize $\scriptsize \mathbf{0.89}$ &\scriptsize $0.87$   \\ 
		\hline
		\scriptsize Brisque &\scriptsize $\mathbf{35.2}$ &\scriptsize $40.2$ &\scriptsize $37.1$ &\scriptsize $42.9$   \\
		\hline
		\scriptsize Time (min) &\scriptsize 47 &\scriptsize 13 &\scriptsize 17 &\scriptsize 57   \\ 
		\hline
	\end{tabular} 
\caption{Quantitative comparisons based on six datasets: We used six datasets to obtain the average SER, PSNR, SSIM, Brisque score, and time used for reconstruction.}
\label{quan_comp3}
\end{table}

In Fig. \ref{comp22}, we compare the region of interest for SToRM500, SToRM with 150 frames (SToRM150), the proposed method with two different $d$ values, the unsupervised Time-DIP approach, and the low-rank algorithm. From Fig. \ref{comp22}, we observe that the proposed scheme can significantly reduce errors in comparison to SToRM150. Additionally, the proposed scheme is able to capture the motion patterns better than Time-DIP, while the low-rank method is unable to capture the motion patterns. From the time profile in Fig. \ref{comp22}, we notice that the proposed scheme is capable of recovering the abrupt change in blood-pool contrast between diastole and systole. This is due to inflow effects associated with gradient echo (GRE) acquisitions. In particular, the blood from regions outside the slice enters the heart, which did not experience any of the former slice-selective excitation pulses; the differences in magnetization of the blood with no magnetization history, and that was within the slice, results in the abrupt change in intensity. We note that some of the competing methods such as Time-DIP and low-rank, blur these details.  

\begin{figure*}[!h]
	\centering
	\subfigure[Visual comparisons]{\includegraphics[width=0.75\textwidth]{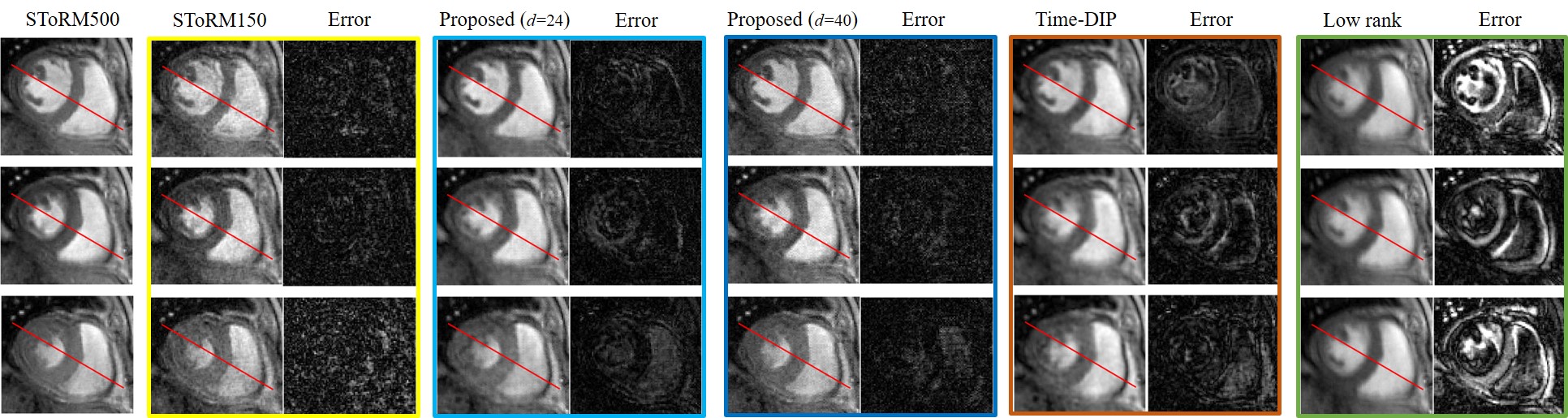}}\quad
	\subfigure[Time profiles]{\includegraphics[width=0.22\textwidth]{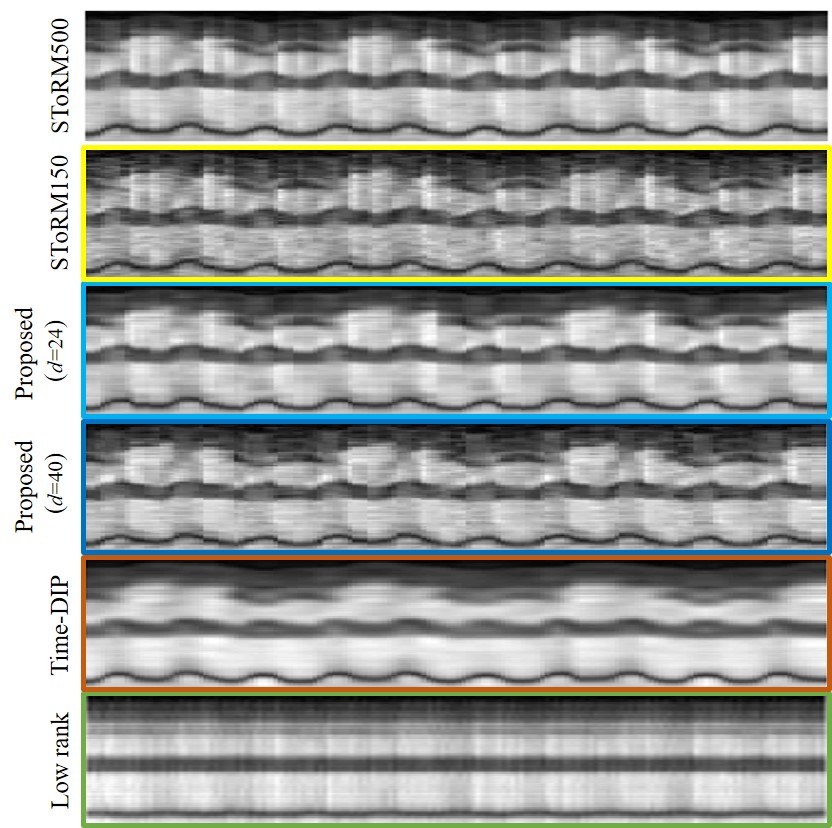}}
	\caption{Comparisons with the state-of-the-art methods. The first column of (a) corresponds to the reconstructions from 500 frames ($\sim$ 25s of acquisition time), while the rest of the columns are recovered from 150 frames ($\sim$ 7.5s of acquisition time). The top row of (a) corresponds to the diastole phase, while the third row is the diastole phase. The second row of (a) is an intermediate one. Fig. (b) corresponds to the time profiles of the reconstructions. We observe that the proposed ($d=40$) reconstructions exhibit less blurring and fewer artifacts when compared to SToRM150 and competing methods.}
	\label{comp22}
\end{figure*}

We also perform the comparisons on a different dataset in Fig. \ref{comp3}. We compare the proposed scheme with SToRM500, SToRM150, Time-DIP, and the low-rank approach. The results are shown in Fig. \ref{comp3}. From the figure, we see that the proposed reconstructions appear less blurred than those of the  conventional schemes.

\begin{figure*}[!h]
	\centering
	\subfigure[Visual comparisons]{\includegraphics[width=0.7\textwidth]{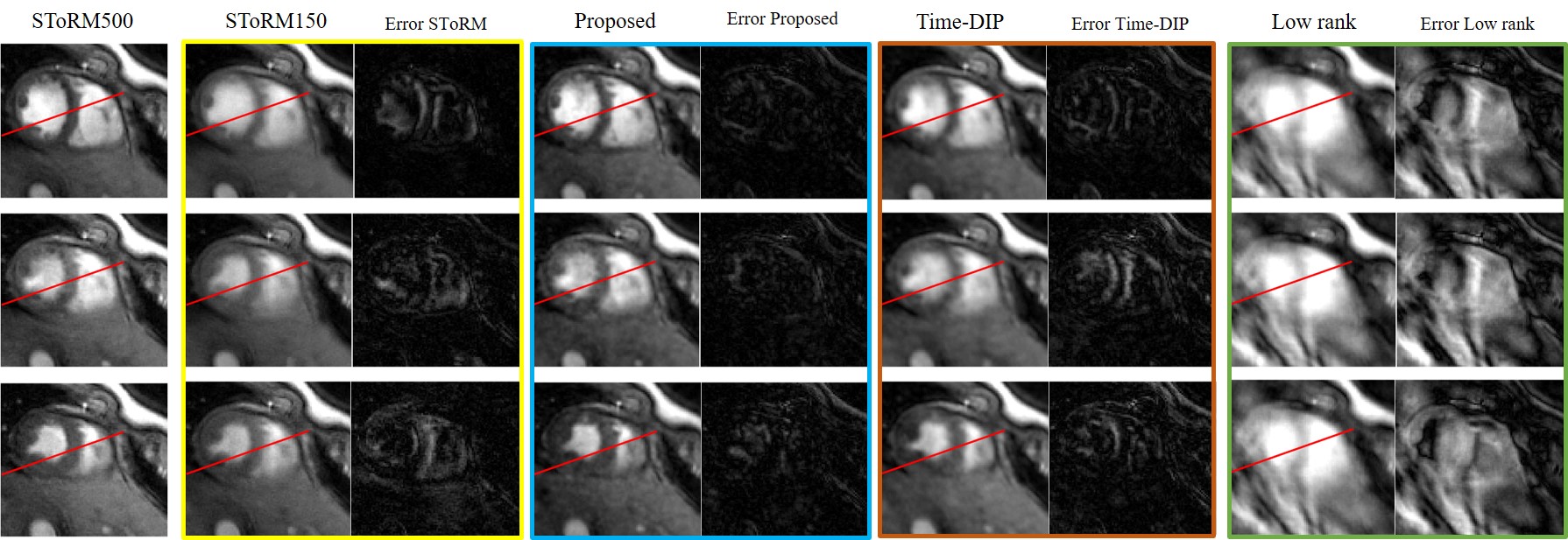}}\quad
	\subfigure[Time profiles]{\includegraphics[width=0.22\textwidth]{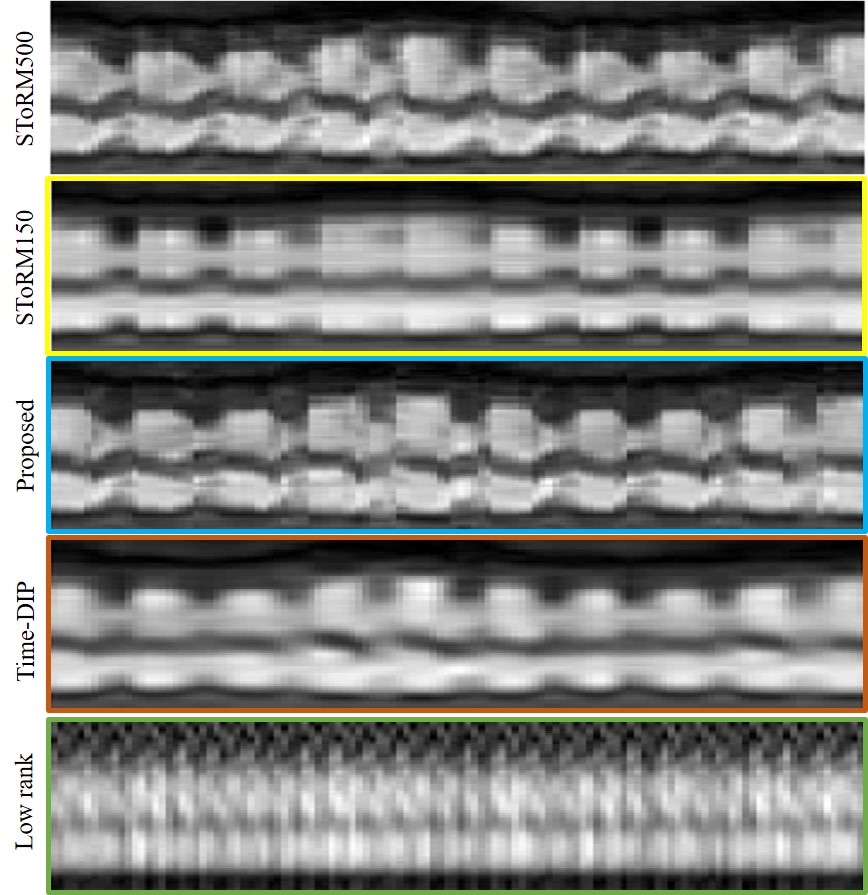}}
	\caption{Comparisons with the state-of-the-art methods. The first column of (a) corresponds to the reconstructions from 500 frames ($\sim$ 25s of acquisition time), while the rest of the columns are recovered from 150 frames ($\sim$ 7.5s of acquisition time). The top row of (a) corresponds to the diastole phase, while the third row is the diastole phase. The second row of (a) is an intermediate one. Fig. (b) corresponds to the time profiles of the reconstructions. We chose $d=40$ for the proposed scheme. We observe that the proposed reconstructions appear less blurred when compared to the  conventional schemes.}
	\label{comp3}
\end{figure*}

\begin{figure*}[!h]
	\centering
	\subfigure[Latent vectors]{\includegraphics[width=0.2\textwidth]{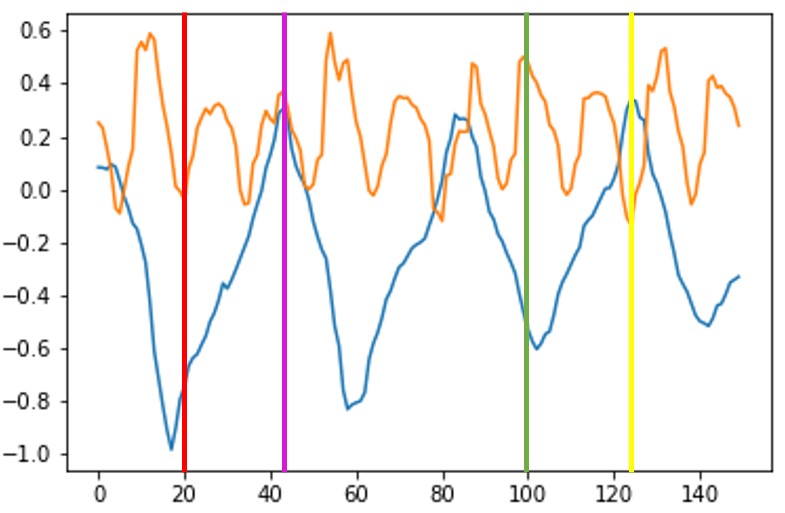}}\hspace{1em}
	\subfigure[Systole in E-E]{\includegraphics[width=0.13\textwidth]{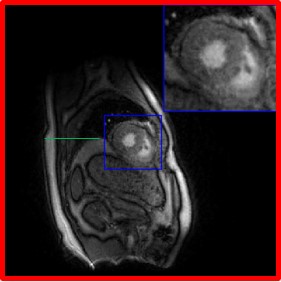}}\hspace{1em}
	\subfigure[Systole in E-I]{\includegraphics[width=0.13\textwidth]{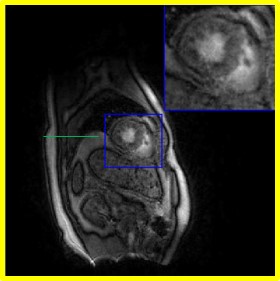}}\hspace{1em}
	\subfigure[Diastole in E-E]{\includegraphics[width=0.13\textwidth]{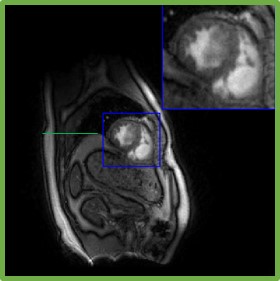}}\hspace{1em}
	\subfigure[Diastole in E-I]{\includegraphics[width=0.13\textwidth]{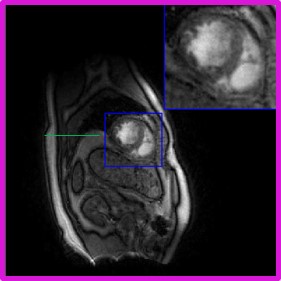}}
	\caption{Illustration of the framework of the proposed scheme with $d=40$. We plot the latent variables of 150 frames in a time series on the first dataset. We showed four different phases in the time series: systole in End-Expiration (E-E), systole in End-Inspiration (E-I), diastole in End-Expiration (E-E), and diastole in End-Inspiration (E-I). A thin green line surrounds the liver in the image frame to indicate the respiratory phase. The latent vectors corresponding to the four different phases are indicated in the plot of the latent vectors.}
	\label{showcase1}
\end{figure*}

\begin{figure*}[!h]
	\centering
	\subfigure[Latent vectors]{\includegraphics[width=0.2\textwidth]{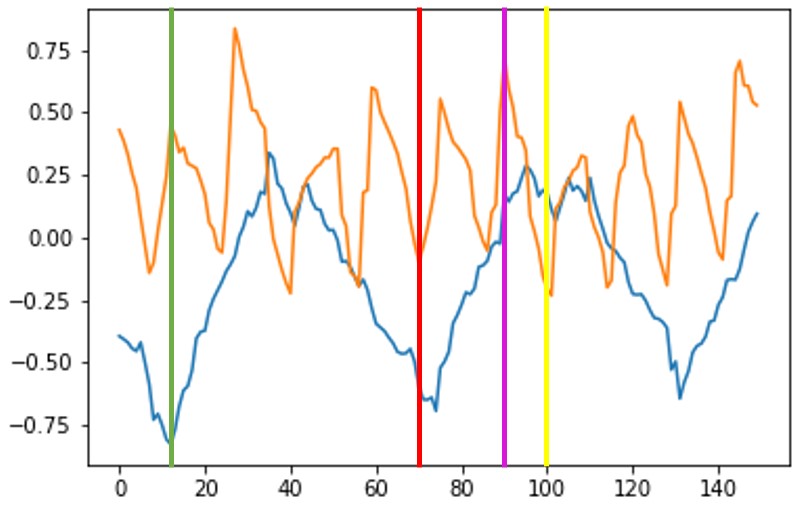}}\hspace{1em}
	\subfigure[Systole in E-E]{\includegraphics[width=0.13\textwidth]{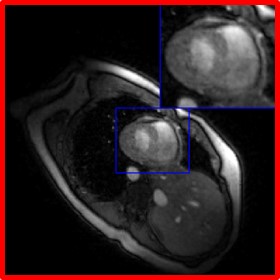}}\hspace{1em}
	\subfigure[Systole in E-I]{\includegraphics[width=0.13\textwidth]{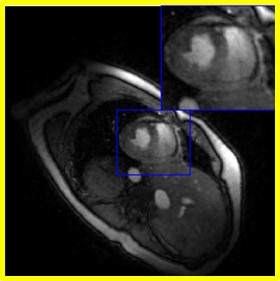}}\hspace{1em}
	\subfigure[Diastole in E-E]{\includegraphics[width=0.13\textwidth]{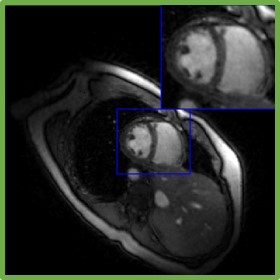}}\hspace{1em}
	\subfigure[Diastole in E-I]{\includegraphics[width=0.13\textwidth]{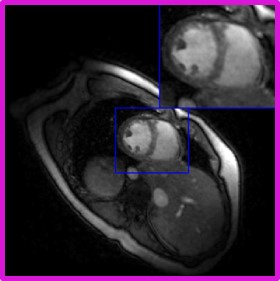}}
	\caption{Illustration of the framework of the proposed scheme with $d=40$. We plot the latent variables of 150 frames in a time series. We showed four different phases in the time series: systole in End-Expiration (E-E), systole in End-Inspiration (E-I), diastole in End-Expiration (E-E), and diastole in End-Inspiration (E-I). The latent vectors corresponding to the four different phases are indicated in the plot of the latent vectors.}
	\label{showcase2}
\end{figure*}

We also compared the proposed scheme with SToRM500, SToRM150, and the unsupervised Time-DIP approach quantitatively. We omit the low-rank method here because low-rank approach often failed in some datasets. The quantitative comparisons are shown in Table \ref{quan_comp3}. We used SToRM500 as the reference for SER, PSNR, and SSIM calculations. The quantitative results are based on the average performance from six datasets.

Finally, we illustrate the proposed approaches in Fig. \ref{showcase1} and Fig. \ref{showcase2}, respectively. The proposed approach decoupled the latent vectors corresponding to the cardiac and respiratory phases well, as shown in the representative examples in Fig. \ref{showcase1} (a) and  Fig. \ref{showcase2} (a).


\section{Conclusion}

In this work, we introduced an unsupervised generative SToRM framework for the recovery of free-breathing cardiac images from spiral acquisitions. This work assumes that the images are generated by a non-linear CNN-based generator $\mathcal{G}_{\theta}$, which maps the low-dimensional latent variables to high-resolution images. Unlike traditional supervised CNN methods, the proposed approach does not require any training data. The parameters of the generator and the latent variables are directly estimated from the undersampled data. The key benefit for this generative model is its ability to compress the data, which results in a memory-effective algorithm. To improve the performance, we introduced a network/distance regularization and a latent variable regularization. The combination of the priors ensures the learning of representations that preserve distances and ensure the temporal smoothness of the recovered images; the regularized approach provides improved reconstructions while minimizing the need for early stopping. To reduce the computational complexity, we introduced a fast approximation of the data loss term as well as a progressive training-in-time strategy. These approximations result in an algorithm with computational complexity comparable to our prior SToRM algorithm. The main benefits of this scheme are the improved performance and considerably reduced memory demand. While our main focus in this work was to establish the benefits of this work in 2D, we plan to extend this work to 3D applications in the future.

\section*{Acknowledgement}

The authors would like to thank Ms. Melanie Laverman from the University of Iowa for making editorial corrections to refine this paper.

\bibliographystyle{IEEEbib}
\bibliography{refs}

\end{document}